\documentclass[12pt, a4paper]{article}

\usepackage{amssymb}
\usepackage{amsthm,color}
\usepackage{latexsym}
\usepackage{graphicx}
\usepackage{fancyhdr}
\usepackage{bbm}
\usepackage{amsmath}
\usepackage{latexsym,amssymb,amsxtra,mathrsfs,times}

\renewcommand{\vec}[1]{\underline{#1}}

\newtheorem{thm}{Theorem}

\newtheorem{lemma}[thm]{Lemma}
\newtheorem{cor}[thm]{Corollary}

\theoremstyle{definition}

\newtheorem{exam}{Example}

\newcommand{\tr}{{\mathrm{Tr}}}
\newcommand{\gf}{{\mathrm{GF}}}

\newcommand{\ca}{{\mathcal{C}_{(d_0,d_1,\cdots,d_t)}^{(1)}}}
\newcommand{\cb}{{\mathcal{C}_{(\widetilde{d}_1,\cdots,\widetilde{d}_t)}^{(2)}}}
\newcommand{\catwo}{{\mathcal{C}_{(d_0,d_1)}^{(1)}}}
\newcommand{\cbone}{{\mathcal{C}_{(\widetilde{d}_1)}^{(2)}}}
\newcommand{\cbtwo}{{\mathcal{C}_{(\widetilde{d}_1,\widetilde{d}_2)}^{(2)}}}

\begin{document}

\title{Optimal cyclic codes with generalized Niho type zeroes and the weight distribution}

\author{
Maosheng Xiong\thanks{Department of Mathematics, The Hong Kong University of Science and Technology, Clear Water Bay, Kowloon, Hong
 Kong, China (e-mail: mamsxiong@ust.hk).} \, and \,Nian Li\thanks{Department of Informatics, University of Bergen, 5020 Bergen, Norway (e-mail: nianli.2010@gmail.com).}}

\maketitle

%\slugger{mms}{xxxx}{xx}{x}{x--x}%slugger should be set to mms, siap, sicomp, %sicon, sidma, sima, simax, sinum, siopt, sisc, or sirev

\begin{abstract}
\footnotesize{
%In a recent work \cite{XLZD} we obtained the weight distribution of two classes of cyclic codes $\ca$ and $\cb$ whose duals may have arbitrary number of generalised Niho-type zeroes.

In this paper we extend the works \cite{gegeng2,XLZD} further in two directions and compute the weight distribution of these cyclic codes under more relaxed conditions. It is interesting to note that many cyclic codes in the family are optimal and have only a few non-zero weights. Besides using similar ideas from \cite{gegeng2,XLZD}, we carry out some subtle manipulation of certain exponential sums.

\vspace*{2mm}

\vspace*{2mm}
}

\end{abstract}

\section{Introduction}\label{sec-into}

Cyclic codes are an important class of linear codes. Due to their desirable algebraic properties and efficient algorithms for encoding and decoding processes, cyclic codes have been widely used in many areas such as communication and data storage system. They can also be used to construct other interesting structures such as quantum codes \cite{qc}, frequency hopping sequences \cite{fhs} and so on.

Let $p$ be a prime number, $l \ge 1, q=p^l$ and $\gf(q)$ be the finite field of order $q$. A cyclic code $\mathcal{C}$ of length $n$ over $\gf(q)$ (assume $(n,q)=1$), by the one-to-one correspondence
$$\begin{array}{cccl}
\sigma:& \mathcal{C}&\rightarrow &R:=\gf(q)[x]/(x^n-1)\\
 &(c_0,c_1,\cdots ,c_{n-1})&\mapsto&c_0+c_1x+\cdots +c_{n-1}x^{n-1},
\end{array}$$
can be identified with an ideal of $R$. There exists a unique monic polynomial $g(x)$ with least degree such that $\sigma(\mathcal{C})=g(x)R$ and $g(x)\mid (x^n-1)$. The $g(x)$ is called the \textit{generator polynomial} and $h(x):=(x^n-1)/g(x)$ is called the \textit{parity-check polynomial} of $\mathcal{C}$. If $h(x)$ has $t$ irreducible factors over $\gf(q)$, we follow the literature and say that ``the dual of $\mathcal{C}$ has $t$ zeroes''. (Note that this is different from \cite{XLZD} in which we call ``$\mathcal{C}$ has $t$ zeroes'' instead.) $\mathcal{C}$ is called \emph{irreducible} if $t=1$ and \emph{reducible} if $t \ge 2$.

Denote by $A_i$ the number of codewords of $\mathcal{C}$ with Hamming weight $i$, where $0 \le i \le n$. The study of the weight distribution $(A_0,A_1,\cdots ,A_n)$ or equivalently the weight enumerator $1+A_1Y+A_2Y^2+\cdots+A_nY^n$ is important in both theory and application, because the weight distribution gives the minimum distance and thus the error correcting capability of the code, and the weight distribution allows the computation of the probability of error detection and correction with respect to some algorithms \cite{Klov}. Moreover, the weight distribution is related to interesting and challenging problems in number theory (\cite{cal,Schroof}).

In a recent paper \cite{gegeng2}, the authors constructed some classes of cyclic codes whose duals have two Niho type zeroes and obtained the weight distribution. These classes of cyclic codes are quite interesting because they contain some optimal cyclic codes (among linear codes) and in general have only three or four non-zero weights. This beautiful work was recently extended to more general classes of cyclic codes whose duals may have arbitrary number of Niho type zeroes (see \cite{XLZD}). The purpose of this paper is to extend these works in yet two other directions. It is interesting to note that this not only vastly generalizes the construction of \cite{gegeng2,XLZD}, but also yields many optimal or almost optimal cyclic codes with very few non-zero weights, none of which was present in \cite{gegeng2,XLZD} (See Examples \ref{1:ex1}--\ref{2:ex2}, Tables \ref{Table:ex1} and \ref{Table:ex2} in Section \ref{sec-2}). We study the weight distribution of these cyclic codes, by employing similar ideas from \cite{gegeng2,XLZD} and by carrying out some quite subtle analysis of certain exponential sums.

In recent years, for many families of cyclic codes the weight distribution problem has been solved. We only mention here that most of the results are for cyclic codes whose duals have no more than three zeroes (see for example \cite{AL06,BM72,BM73,fit,McE72,Rao10,schmidt,van,Vega1,Vega2,wol,D-Y12} and \cite{Ding2,FL08,Feng12,F-M12,holl,luo2,luo3,luo4,Ding1,M09,Mois09,
Vega12,Tang12,Xiong1,Xiong2,Xiong3,zeng}). There are only a few results for cyclic codes whose duals may have arbitrary number of zeroes (\cite{gegeng, Y-X-D12,Y-X-X14,XLZD}). The duals of the cyclic codes considered in this paper may also have arbitrary number of zeroes.

The paper is organized as follows. In Section \ref{sec-2}, for any prime $p$, we introduce the cyclic codes $\ca$ and $\cb$ and the main results (Theorems \ref{1:thm1} and \ref{1:thm2}). The special cases of 3-weight and 4-weight cyclic codes are presented in Corollaries \ref{cor:ca} and \ref{cor:cb}. Then we provide numerical examples of optimal or almost optimal cyclic codes over $\gf(4)$ and $\gf(8)$ and compute the weight distribution. We also compile a list of such codes over other finite fields. In Section \ref{sec-00} we prove a simple lemma which will be used later. For $p=2$, Statements (i) of Theorems \ref{1:thm1} and \ref{1:thm2} are proved in Section \ref{sec-3}, and Statements (ii) of Theorems \ref{1:thm1} and \ref{1:thm2} are proved in Section \ref{sec-4}. For $p \ge 3$, Statements (i) and (ii) of Theorems \ref{1:thm1},\ref{1:thm2} are proved in Sections \ref{sec-32} and \ref{sec-42} respectively. As was noted in \cite{gegeng2,XLZD}, we remark here that the proofs for $p=2$ and $p \ge 3$ are quite different. In Section \ref{sec-conclusion} we conclude the paper.

%$p \ge 3$ and we introduce the cyclic codes $\caq$ and $\cbq$ and the main results (Theorems \ref{1:thm3} and \ref{1:thm4}). The special cases of 3-weight and 4-weight cyclic codes are presented in Corollaries \ref{cor:cap} and \ref{cor:cbp}. We also give examples of computing the weight distribution and compile a list of optimal cyclic codes in these families at the end of the section.

\section{Cyclic codes $\ca$ and $\cb$} \label{sec-2}

%We consider exponents of the form $d \equiv \triangle \pmod{q-1}$. Note that $d$ is of Niho type if $\triangle=p^j$ for some integer $j$, thus we may call exponents of this form {\em generalized Niho exponents}.

%We see that $d_j \equiv 2 \triangle \pmod{q-1}$ are generalized Niho exponents.

From what follows, let $p$ be a prime number, and $l,m$ be positive integers. Let $q=p^l$ and $r=q^m$. Denote by $\gf(r^2)$ the finite field of order $r^2$. Let $\gamma$ be a primitive element of $\gf(r^2)$.

%A positive integer $d$ is called a {\em Niho exponent} if $d \equiv p^j\pmod{q-1}$ for some $j$. The Niho exponents were originally introduced by Niho \cite{Niho-PhD} who investigated the cross-correlation between an $m$-sequence and its decimation. Since then, Niho exponents were further studied and had been used in other research topics. %For cyclic codes whose duals have two or three zeroes of Niho type, the reader is referred to \cite{Charpin,Li-Zeng-Hu} and the recent work \cite{gegeng2}.

\subsection{Cyclic code $\ca$}
For any integers $h,f$, define
\begin{eqnarray} \label{2:e1} e=(h,r+1), \quad \delta=\left((r+1)f,(r-1)e\right), \quad n=(r^2-1)/\delta. \end{eqnarray}
Let $t$ be an integer and assume that
\begin{itemize}
\item[(a).] $1 \le t < \frac{r+1}{2e}$,

\item[(b).] $\left(f,\frac{r-1}{q-1}\right)=1$,

\item[(c).] if $p$ is odd, then $m$ is odd, or $m$ and $h$ are both even.
\end{itemize}
Let $d_0,d_1,\ldots,d_t$ be integers such that
\begin{eqnarray} \label{2:da1}
d_j \equiv (jh+f) (r-1)+2 f  \pmod{r^2-1}, \quad 0 \le j \le t\, .
\end{eqnarray}

A positive integer $d$ is called a {\em Niho exponent} if $d \equiv q^j\pmod{r-1}$ for some $j$. The Niho exponents were originally introduced by Niho \cite{Niho-PhD} who investigated the cross-correlation between an $m$-sequence and its decimation. Since then, Niho exponents were further studied and had been used in other research topics. Here $d_j \equiv 2f \pmod{r-1}\, \forall j$ and $\left(f,\frac{r-1}{q-1}\right)=1$, the $d_j$'s are called the ``generalized Niho exponents'', and the ${\gamma}^{-d_j}$'s are called the ``generalized Niho type zeroes''. %For cyclic codes whose duals have two or three zeroes of Niho type, the reader is referred to \cite{Charpin,Li-Zeng-Hu} and the recent work \cite{gegeng2}.

It can be seen that $(d_0,d_1,\ldots,d_t,r^2-1)=\delta$. The $q$-ary cyclic code $\ca$ of length $n$ consists of elements $c(\vec{a})$ given by
\begin{equation}\label{2:ca}
\begin{array}{l}
c(\vec{a})=\left(\tr_{r/q}\left(a_0 \gamma^{d_0 i}\right)+\tr_{r^2/q}\left(\sum_{j=1}^t a_j \gamma^{d_j i}  \right)\right)_{i=0}^{n-1},
\end{array}\end{equation}
where $\vec{a}=(a_0,a_1,\ldots,a_t)$ for any $a_1,\cdots,a_{t}\in \gf(r^2)$ and $a_0 \in \gf(r)$. Here $\tr_{r/q}$ and $\tr_{r^2/q}$ denote the trace map from $\gf(r)$ and $\gf(r^2)$ to $\gf(q)$ respectively. It will be seen that the dimension of $\ca$ is always $(2t+1)m$ and the dual of $\ca$ has the $(t+1)$ zeroes ${\gamma}^{-d_0},\ldots, {\gamma}^{-d_t}$.

We remark that a similar $\ca$ over $\gf(p)$ was constructed in \cite{XLZD}, but it required $(2f,r-1)=1$, which is valid only when $p=2$. Here we consider $\ca$ over $\gf(q)$ for any $p$ under more flexible conditions. Note that (b) reduces to $(2f,r-1)=1$ only if $q=p=2$. As it turns out, the weight distribution of the new $\ca$ is very similar to that in \cite{XLZD}. However, the proofs are much more involved. For the sake of completeness, we describe the weight distribution as follows.

Let $N_0=1,N_1=0$ and define
\begin{eqnarray} \label{2:nr}
N_k:=k! e^k \sum_{\substack{\lambda_2,\lambda_3,\ldots,\\
\sum_{j \ge 2} j \lambda_j=k}} \binom{\frac{r+1}{e}}{\sum_{j } \lambda_j} \left(\sum_{j } \lambda_j\right)! \prod_{j } \frac{\left(B_j/j!\right)^{\lambda_j}}{(\lambda_j)!},\quad \forall \, k \ge 2. \end{eqnarray}
Here the summation is over all non-negative integers $\lambda_2,\lambda_3,\ldots$ such that $\sum_{j \ge 2}j \lambda_j=k$ and
\begin{eqnarray*}
 B_j:=r^{-1}(r-1)^j+(-1)^j(1-r^{-1}),\end{eqnarray*}
and $\binom{u}{v}$ is the standard binomial coefficient ``$u$-choose-$v$''. It is easy to compute that $N_2=e(r^2-1),N_3=e^2(r-2)(r^2-1)$, $N_4=e^2(r^2-1)\left\{(e+3)r^2-6er+6e-3\right\}$, etc. We prove the following.

\begin{thm} \label{1:thm1}
(i). For $p=2$ or $p$ being an odd prime, under assumptions (\ref{2:e1})--(\ref{2:ca}) and (a)--(c), the code $\ca $ is a $q$-ary cyclic code of length $n=(r^2-1)/\delta$ and dimension $(2t+1)m$, with at most $(2t+1)$ non-zero weights, each of which is given by
    \begin{eqnarray} \label{2:wj} w_j=\frac{q-1}{q \delta} \cdot \left(r^2-(je-1)r\right), \, 0 \le j \le 2t. \end{eqnarray}
(ii). Let $\mu_j$ be the frequency of the weight $w_j$ for each $j$. Define $\vec{\mu}=(\mu_0,\mu_1,\ldots,\mu_{2t})^T$, and  $\vec{b}=(b_0,b_1,\ldots,b_{2t})^T$ where $b_i=r^{2t+1}N_i-\left(r^2-1\right)^i$. Then
    \[\vec{\mu}=\left(M^{(1)}_t\right)^{-1}\vec{b}. \]
Here $M_t^{(1)}=[m_{ij}]_{0 \le i,j \le 2t}$ is an invertible Vandermonde matrix whose entry is given by $m_{ij}=\left(jer-r-1\right)^i$.
\end{thm}

By using computer algebra such as {\bf Mathematica}, Theorem \ref{1:thm1} can be used easily to compute the weight distribution of $\ca$ explicitly for any $t$, though the results are quite complicated to be written down even for $t =2$. When $t=1$ which is the most interesting, it is a 3-weight cyclic code, whose weight distribution can be described as follows.

\begin{cor} \label{cor:ca} Under assumptions (\ref{2:e1})--(\ref{2:ca}) and (a)--(c) for $t=1$, the code $\catwo$ is a 3-weight cyclic code of length $n=(r^2-1)/\delta$ and dimension $3m$. The weight distribution is given by Table \ref{Table1}.
\end{cor}

\begin{table}[ht]
\caption{The weight distribution of $\catwo$}\label{Table1}
\begin{center}{
\begin{tabular}{|l|l|}
  \hline
  Weight & Frequency  \\
  \hline
  \hline
$0$ & once \\
  \hline
  $\frac{q-1}{q \delta} \left(r^2+r\right)$ & $\frac{-1+3e - 2 e^2 - q + 2 e q + r^2 - 3 e r^2 + r^3 -
    2 e r^3 + 2 e^2 r^3}{2e^2}$\\
  \hline
  $\frac{q-1}{q \delta} \left(r^2-(e-1)r\right)$ & $\frac{1 - 2 e + r - e r - r^2 + 2 e r^2 - r^3 + e r^3}{e^2}$\\
  \hline
  $\frac{q-1}{q \delta} \left(r^2-(2e-1)r\right)$ & $\frac{-1 + e - r + r^2 - e r^2 + r^3}{2e^2}$\\
  \hline
\end{tabular}}
\end{center}
\end{table}

From what follows we present some interesting examples of cyclic codes from $\ca$ over $\gf(4)$ and $\gf(8)$ respectively which we find optimal or almost optimal by checking ``Bounds on the minimum distance of linear codes'' provided by the website http://www.codetables.de/. Note that these cyclic codes have only a few none-zero weights. Here we omit optimal cyclic codes which could be obtained from \cite{gegeng2,XLZD} and hence only consider the case that $(f,r-1) >1$.

\begin{exam} \label{1:ex1} Let $q=4,m=2,r=16,h=1,f=3$. Then $e=1,(f,r-1)=3$.
\begin{itemize}
\item[(1).] $t=1$: $(d_0,d_1)=(51,66)$. Both Theorem \ref{1:thm1} and numerical computation by {\bf Magma} show that $\mathcal{C}_{(51,66)}^{(1)}$ is a three-weight cyclic code with the weight enumerator
\[1+2040Y^{60}+255Y^{64}+1800Y^{68}.\]
This is a $[85,6,60]$ code over $\gf(4)$ which is optimal among linear codes.

\item[(2).] $t=2$: $(d_0,d_1,d_2)=(51,66,81)$. Both Theorem \ref{1:thm1} and numerical computation by {\bf Magma} show that $\mathcal{C}_{(51,66,81)}^{(1)}$ is a five-weight cyclic code with the weight enumerator
\[1+35700Y^{52}+30600Y^{56}+250920Y^{60}+377655Y^{64}+353700Y^{68}.\]
This is a $[85,10,52]$ code over $\gf(4)$. It is known that for optimal linear codes of length 85 and dimension 10 over $\gf(4)$, the minimal distance satisfies $52 \le d \le 56$.

\item[(3).] $t=3$: $(d_0,d_1,d_2,d_3)=(51,66,81,96)$. Both Theorem \ref{1:thm1} and numerical computation by {\bf Magma} show that $\mathcal{C}_{(51,66,81,96)}^{(1)}$ is a seven-weight cyclic code with the weight enumerator
$1+185640Y^{44}+464100Y^{48}+4641000Y^{52}+17646000Y^{56}+
54396600Y^{60}+101483115Y^{64}+89619000Y^{68}$. This is a $[85,14,44]$ code over $\gf(4)$. It is known that for optimal linear codes of length 85 and dimension 14 over $\gf(4)$, the minimal distance satisfies $48 \le d \le 53$.

\end{itemize}
\end{exam}

\begin{exam} \label{1:ex2} Let $q=8,m=1,r=8,h=1,f=7,r^2=64$. Then $e=1,(f,r-1)=7$.
\begin{itemize}
\item[(1).] $t=1$: $(d_0,d_1)=(63,70)$. Both Theorem \ref{1:thm1} and numerical computation by {\bf   Magma} show that $\mathcal{C}_{(63,70)}^{(1)}$ is a three-weight cyclic code with the weight enumerator
\[1+252Y^{7}+63Y^{8}+196Y^{9}.\]
This is a $[9,3,7]$ code over $\gf(8)$ which is optimal among linear codes.
\item[(2).] $t=2$: $(d_0,d_1,d_2)=(63,70,77)$. Both Theorem \ref{1:thm1} and numerical computation by {\bf Magma} show that $\mathcal{C}_{(63,70,77)}^{(1)}$ is a five-weight cyclic code with the weight enumerator
\[1+882Y^{5}+1764Y^{6}+7812Y^{7}+12411Y^8+9898Y^9.\]
This is a $[9,5,5]$ code over $\gf(8)$ which is optimal among linear codes.
\item[(3).] $t=3$: $(d_0,d_1,d_2,d_3)=(63,70,77,84)$. Both Theorem \ref{1:thm1} and numerical computation by {\bf Magma} show that $\mathcal{C}_{(63,70,77,84)}^{(1)}$ is a seven-weight cyclic code with the weight enumerator
\[1+588Y^{3}+4410Y^{4}+33516Y^{5}+154056Y^{6}+463428Y^7+810621Y^8+630532Y^9.\]
This is a $[9,7,3]$ code over $\gf(8)$ which is optimal among linear codes.
\end{itemize}
\end{exam}

Now we present a table of cyclic codes from $\ca$ over $\gf(3),\gf(9),\gf(5)$ and $\gf(7)$ respectively which we find optimal or almost optimal by checking ``Bounds on the minimum distance of linear codes'' provided by the website http://www.codetables.de/. For simplicity, only the parameters of the codes are listed. None of these cyclic codes can be obtained from \cite{gegeng2,XLZD}.

\begin{table}[ht]
\caption{Optimal or almost optimal cyclic codes from $\ca$}\label{Table:ex1}
\begin{center}{
\begin{tabular}{|l|l|c|c|c|}
  \hline
  $q$ & $r$ & $(t=,h=,f=)$ & code parameters & Optimal (?)\\
  \hline
  \hline
  $3$ & $27$ & $(1,2,1)$ & $[182,9,108]$ & $111 \le d \le 115$ is optimal \\
  \hline
  $9$ & $9$ & $(1,1,2)$ & $[20,3,16]$ & $16 \le d \le 17$ is optimal\\
  \hline
  $9$ & $9$ & $(1,1,4)$ & $[10,3,8]$ &Y\\
  \hline
  $9$ & $9$ & $(2,1,4)$ & $[10,5,6]$ &Y\\
  \hline
  $9$ & $9$ & $(3,1,4)$ & $[10,7,4]$ &Y\\
  \hline
  $9$ & $9$ & $(4,1,4)$ & $[10,9,2]$ &Y\\
  \hline
  $9$ & $9$ & $(1,2,8)$ & $[5,3,3]$ &Y\\
  \hline
  \hline
  $5$ & $5$ & $(1,1,1)$ & $[12,3,8]$ &Y\\
  \hline
  $5$ & $5$ & $(1,1,2)$ & $[6,3,4]$ &Y\\
  \hline
  $5$ & $5$ & $(2,1,2)$ & $[6,5,2]$ &Y\\
  \hline
  \hline
  $7$ & $7$ & $(1,1,2)$ & $[24,3,18]$ &$d=19$ is optimal\\
  \hline
  $7$ & $7$ & $(1,1,3)$ & $[8,3,6]$ &Y\\
  \hline
  $7$ & $7$ & $(2,1,3)$ & $[8,5,4]$ &Y\\
  \hline
  $7$ & $7$ & $(3,1,3)$ & $[8,7,2]$ &Y\\
  \hline
  $7$ & $7$ & $(1,2,3)$ & $[4,3,2]$ &Y\\
  \hline
%  \hline
\end{tabular}}
\end{center}
\end{table}

%The ideas of proofs of Theorems 1 and 2 are similar to those of \cite[Theorem 1]{XLZD}, but the details are much more subtle and require some special care. %We also remark that the conditions (b)(c) of (i) and (ii) seem ``optimal'' in the sense that if the parameters do not satisfy such conditions, then the number of non-zero weights is in general much larger than those described by Theorem 1.

\subsection{Cyclic code $\cb$}
For any integers $h,f$ and $t \ge 1$, define
\begin{eqnarray} \label{2:e2} e=(h,r+1), \quad \delta=\left\{\begin{array}{ll}
\left(\frac{h+f}{2}(r-1)+f,r^2-1\right) & \mbox{ if } t=1\\
\left(\frac{h+f}{2}(r-1)+f,(r-1)e\right) & \mbox{ if } t \ge 2 \end{array} \right., \quad n=(r^2-1)/\delta. \end{eqnarray}
Assume that
\begin{itemize}
\item[(a').] $1 \le t \le \frac{r+1}{2e}$,

\item[(b').] if $p=2$, then $\left(f,\frac{r-1}{q-1}\right)=1$,

\item[(c').] if $p \ge 3$, then

\begin{itemize}
\item[(c1').] $h \equiv f \pmod{2}$ and  $\left(f,\frac{r-1}{q-1}\right)=1$, or

\item[(c2').] $h \equiv f \equiv 0 \pmod{2}$ and  $\left(\frac{f}{2},\frac{r-1}{q-1}\right)=1$.

\end{itemize}

\end{itemize}
Let $\widetilde{d}_1,\ldots,\widetilde{d}_t$ be integers such that
\begin{eqnarray} \label{2:db}
\widetilde{d}_j \equiv  \left(j \cdot h+\frac{f-h}{2}\right)(r-1)+f \pmod{r^2-1}, \quad 1 \le j \le t.
\end{eqnarray}
Here if $p=2$, the number $\frac{1}{2}$ shall be interpreted as an integer which is the multiplicative inverse of $2 \pmod{r-1}$. It can be seen that $(\widetilde{d}_1,\ldots,\widetilde{d}_t,r^2-1)=\delta$. The $q$-ary cyclic code $\cb$ of length $n$ consists of elements $\widetilde{c}(\vec{a})$ given by
\begin{equation}\label{2:cb}
\begin{array}{l}
\widetilde{c}(\vec{a})=\left(\tr_{r^2/q}\left(\sum_{j=1}^t a_j \gamma^{\widetilde{d}_j i}  \right)\right)_{i=0}^{n-1},
\end{array}\end{equation}
where $\vec{a}=(a_1,\ldots,a_t)$ for any $a_1,\cdots,a_{t}\in \gf(r^2)$. Note that the dual of $\cb$ has the $t$ zeroes $\gamma^{-\widetilde{d}_1},\cdots,\gamma^{-\widetilde{d}_t}$. Here $\widetilde{d}_j \equiv f \pmod{r-1} \, \forall j$, the ${\gamma}^{-\widetilde{d}_j}$'s are call the ``generalized Niho type zeroes''.

We remark that a similar $\cb$ over $\gf(p)$ was constructed in \cite{XLZD} under the condition $(f,r-1)=1$ (see \cite[Theorem 1]{XLZD}). Clearly the conditions (b')(c') are more general and provide more flexible parameters. The weight distribution of $\cb$ can be described as follows.

\begin{thm} \label{1:thm2} (i). For $p=2$ or $p$ being an odd prime, under assumptions (\ref{2:e2})--(\ref{2:cb}) and (a')--(c'), the code $\cb $ is a $q$-ary cyclic code of length $n=(r^2-1)/\delta$ and dimension $2tm$, with at most $2t$ non-zero weights, each of which is given by
    \[\widetilde{w}_j=\frac{q-1}{q\delta} \cdot \left(r^2-(je-1)r\right), \, 0 \le j \le 2t-1. \]
(ii). Let $\widetilde{\mu}_j$ be the frequency of the weight $\widetilde{w}_j$ for each $j$. Define $\vec{\widetilde{\mu}}=(\widetilde{\mu}_0,\widetilde{\mu}_1,\ldots,\widetilde{\mu}_{2t-1})^T$, and $\vec{\widetilde{b}}=(\widetilde{b}_0,\widetilde{b}_1,\ldots,\widetilde{b}_{2t-1})^T$ where $\widetilde{b}_i=r^{2t}N_i-\left(r^2-1\right)^i$. Then
    \[\vec{\widetilde{\mu}}=\left(M^{(2)}_t\right)^{-1}\vec{\widetilde{b}}. \]
Here $M_t^{(2)}=[m_{ij}]_{0 \le i,j \le 2t-1}$ is a $2t \times 2t$ Vandermonde matrix whose entry is given by $m_{ij}:=\left(jer-r-1\right)^i$.
\end{thm}

\begin{cor} \label{cor:cb} (i). Under assumptions (\ref{2:e2})--(\ref{2:cb}) and (a')--(c') for $t=1$, the code $\cbone$ is a $q$-ary cyclic code of length $n=(r^2-1)/\delta$ and dimension $2m$ with most two non-zero weights. The weight distribution is given by Table \ref{Table2}. Note that it is a 1-weight code if and only if $e=1$.

(ii). Under assumptions (\ref{2:e2})--(\ref{2:cb}) and (a')--(c') for $t=2$, the code $\cbtwo$ is a 4-weight cyclic code of length $n=(r^2-1)/\delta$ and dimension $4m$. The weight distribution is given by Table \ref{Table3}.
\end{cor}

Since $\cbone$ is irreducible, (i) of Corollary \ref{cor:cb} should be known to researchers in the field. We collect the result here only for the sake of completeness. However, the dual of $\cbtwo$ has two zeroes, and (ii) of Corollary \ref{cor:cb} is new.

\begin{table}[ht]
\caption{The weight distribution of $\cbone$} \label{Table2}
\begin{center}{
\begin{tabular}{|l|l|}
  \hline
  Weight & Frequency  \\
  \hline
  \hline
$0$ & once \\
  \hline
  $\frac{q-1}{q \delta} \left(r^2+r\right)$ & $\frac{(e-1)(r^2-1)}{e}$\\
  \hline
  $\frac{q-1}{q \delta} \left(r^2-(e-1)r\right)$ & $\frac{r^2-1}{e}$\\
  \hline
\end{tabular}}
\end{center}
\end{table}

\begin{table}[ht]
\caption{The weight distribution of $\cbtwo$}\label{Table3}
\begin{center}{
\begin{tabular}{|l|l|}
  \hline
  Weight & Frequency  \\
  \hline
  \hline
$0$ & once \\
  \hline
  $\frac{q-1}{q \delta} \left(r^2+r\right)$ & $\frac{1 - 6 e + 11 e^2 - 6 e^3 + 2 r - 9 e r + 9 e^2 r +
    3 e r^2 - 5 e^2 r^2 - 2 r^3 + 9 e r^3 - 9 e^2 r^3 -
    r^4 + 3 e r^4 - 6 e^2 r^4 + 6 e^3 r^4}{6e^3}$\\
  \hline
  $\frac{q-1}{q \delta} \left(r^2-(e-1)r\right)$ & $\frac{-1 + 5 e - 6 e^2 - 2 r + 7 e r - 4 e^2 r - 3 e r^2 +
    4 e^2 r^2 + 2 r^3 - 7 e r^3 + 4 e^2 r^3 + r^4 -
    2 e r^4 + 2 e^2 r^4}{2e^3}$\\
  \hline
  $\frac{q-1}{q \delta} \left(r^2-(2e-1)r\right)$ & $\frac{1 - 4 e + 3 e^2 + 2 r - 5 e r + e^2 r + 3 e r^2 -
    3 e^2 r^2 - 2 r^3 + 5 e r^3 - e^2 r^3 - r^4 + e r^4}{2e^3}$\\
  \hline
  $\frac{q-1}{q \delta} \left(r^2-(3e-1)r\right)$ & $\frac{-1 + 3 e - 2 e^2 - 2 r + 3 e r - 3 e r^2 + 2 e^2 r^2 +
    2 r^3 - 3 e r^3 + r^4}{6e^3}$\\
  \hline
\end{tabular}}
\end{center}
\end{table}

From what follows we present some interesting examples of cyclic codes from $\cb$ over $\gf(4)$ and $\gf(8)$ respectively which we find optimal or almost optimal by checking ``Bounds on the minimum distance of linear codes'' provided by the website http://www.codetables.de/. Note that these cyclic codes have only a few none-zero weights. Here we omit optimal cyclic codes which could be obtained from \cite{gegeng2,XLZD} and hence only consider the case that $(f,r-1) >1$.

%In what follows we present some interesting examples of cyclic codes from $\cb$ over $\gf(4)$ and $\gf(8)$ which are optimal or almost optimal among all linear codes. Note that these cyclic codes have only a few none-zero weights. We focus on the case that $(f,r-1) >1$.

\begin{exam} \label{2:ex1}Let $q=4,m=2,r=16,h=2,f=6$. Then $e=1,(f,r-1)=3$.
\begin{itemize}
\item[(1).] $t=1$: $(\widetilde{d}_1)=(66)$. Both Theorem \ref{1:thm2} and numerical computation by {\bf Magma} show that $\mathcal{C}_{(66)}^{(1)}$ is a one-weight cyclic code with the weight enumerator
\[1+255Y^{64}.\]
This is a $[85,4,64]$ code over $\gf(4)$ which is optimal among linear codes.

\item[(2).] $t=2$: $(\widetilde{d}_1,\widetilde{d}_2)=(66,96)$. Both Theorem \ref{1:thm2} and numerical computation by {\bf Magma} show that $\mathcal{C}_{(66,96)}^{(1)}$ is a four-weight cyclic code with the weight enumerator
\[1+10200Y^{56}+4080Y^{60}+30855Y^{64}+20400Y^{68}.\]
This is a $[85,8,56]$ code over $\gf(4)$. It is known that for optimal linear codes of length 85 and dimension 8 over $\gf(4)$, the minimal distance satisfies $56 \le d \le 59$.

\item[(3).] $t=3$: $(\widetilde{d}_1,\widetilde{d}_2,\widetilde{d}_3)=(66,96,126)$. Both Theorem \ref{1:thm2} and numerical computation by {\bf Magma} show that $\mathcal{C}_{(66,96,126)}^{(1)}$ is a six-weight cyclic code with the weight enumerator
$1+92820Y^{48}+142800Y^{52}+1285200Y^{56}+3272160Y^{60}+
6390555Y^{64}+5593680Y^{68}$. This is a $[85,12,48]$ code over $\gf(4)$. It is known that for optimal linear codes of length 85 and dimension 12 over $\gf(4)$, the minimal distance satisfies $48 \le d \le 55$.

\end{itemize}
\end{exam}

\begin{exam} \label{2:ex2} Let $q=8,m=1,r=8,h=2,f=14,r^2=64$. Then $e=1,(f,r-1)=7$.
\begin{itemize}
\item[(1).] $t=1$: $(\widetilde{d}_1)=(70)$. Both Theorem \ref{1:thm1} and numerical computation by {\bf   Magma} show that $\mathcal{C}_{(70)}^{(2)}$ is a one-weight cyclic code with the weight enumerator
\[1+63Y^{8}.\]
This is a $[9,2,8]$ code over $\gf(8)$ which is optimal among linear codes.

\item[(2).] $t=2$: $(\widetilde{d}_1,\widetilde{d}_2)=(70,84)$. Both Theorem \ref{1:thm2} and numerical computation by {\bf Magma} show that $\mathcal{C}_{(70,84)}^{(2)}$ is a four-weight cyclic code with the weight enumerator
\[1+588Y^{6}+504Y^{7}+1827Y^{8}+1176Y^9.\]
This is a $[9,4,6]$ code over $\gf(8)$ which is optimal among linear codes.

\item[(3).] $t=3$: $(\widetilde{d}_1,\widetilde{d}_2,\widetilde{d}_3)=(70,84,98)$. Both Theorem \ref{1:thm2} and numerical computation by {\bf Magma} show that $\mathcal{C}_{(70,84,98)}^{(2)}$ is a six-weight cyclic code with the weight enumerator
\[1+882Y^{4}+3528Y^{5}+19992Y^{6}+57456Y^{7}+101493Y^8+78792Y^9.\]
This is a $[9,6,4]$ code over $\gf(8)$ which is optimal among linear codes.
\end{itemize}
\end{exam}

Now we present a table of cyclic codes from $\cb$ over $\gf(3),\gf(9),\gf(5)$ and $\gf(7)$ respectively which we find optimal or almost optimal by checking ``Bounds on the minimum distance of linear codes'' from the website http://www.codetables.de/. For simplicity, only the parameters of the code are listed. Note that none of these cyclic codes can be obtained from \cite{gegeng2,XLZD}.

\begin{table}[ht]
\caption{Optimal or almost optimal cyclic codes from $\cb$}\label{Table:ex2}
\begin{center}{
\begin{tabular}{|l|l|c|c|c|}
  \hline
  $q$ & $r$ & $(t=,h=,f=)$ & code parameters & Optimal (?)\\
  \hline
  \hline
  $3$ & $27$ & $(1,4,2)$ & $[91,6,54]$ & $57 \le d \le 58$ is optimal \\
  \hline
  \hline
  $9$ & $9$ & $(1,2,8)$ & $[5,2,4]$ & Y \\
  \hline
  $9$ & $9$ & $(2,2,8)$ & $[5,4,2]$ & Y \\
  \hline
  \hline
  $7$ & $7$ & $(1,1,3)$ & $[16,2,14]$ & Y \\
  \hline
  $7$ & $7$ & $(2,1,3)$ & $[16,4,10]$ & $d=11$ is optimal \\
  \hline
%  \hline
\end{tabular}}
\end{center}
\end{table}

%%%%%%%%%%%%%%%%%%%%%%%%%%%%%%%%%%%%%%%%%%%%%%%%%%
%%%%%%%%%%%%%%%%%%%%%%%%%%%%%%%%%%%%%%%%%%%%%%%%%%
\section{Preliminaries}\label{sec-00}
%%%%%%%%%%%%%%%%%%%%%%%%%%%%%%%%%%%%%%%%%%%%%%%%%%
%%%%%%%%%%%%%%%%%%%%%%%%%%%%%%%%%%%%%%%%%%%%%%%%%%

Following notation from Section \ref{sec-2}, let $p$ be a prime, $q=p^l, r=q^m$ and let $\gamma$ be a primitive element of $\gf(r^2)$. Define $\gf(r^2)^*:=\gf(r^2)-\{0\}$. For any integer $d$, let $h_d(x) \in \gf(q)[x]$ be the minimal polynomial of $\gamma^{-d}$ over $\gf(q)$. We first prove the following lemma, the proof of which is similar to \cite[Lemma 5]{XLZD}.

\begin{lemma} \label{3:lem33} Suppose $\left(\triangle,\frac{r-1}{q-1}\right)=1$ or $\left(\frac{\triangle}{2},\frac{r-1}{q-1}\right)=1$ if $\triangle$ is even.
\begin{itemize}
\item[(i).] If $d=s(r-1)+\triangle$, then $\deg h_d(x)=\left\{\begin{array}{lll}
m &:& \mbox{ if } \triangle \equiv 2s \pmod{r+1},\\
2m &:& \mbox{ if } \triangle \not \equiv 2s \pmod{r+1}. \end{array}\right.$

\item[(ii).] If $d=s(r-1)+\triangle$ and $d'=s'(r-1)+\triangle$, then
$h_{d}(x) = h_{d'}(x)$ if and only if $s \equiv s' \pmod{r+1}$ or $s+s' \equiv \triangle \pmod{r+1}$.
\end{itemize}

\end{lemma}

\begin{proof}
(i). $\deg h_d(x)$ is the least positive integer $k$, $1 \le k \le 2m$ such that
$d q^k \equiv d \pmod{r^2-1}$. Since $d \equiv \triangle \pmod{r-1}$, we have $(q^m-1)|\triangle(q^k-1)$. Dividing $q-1$ on both sides, we find that $(q^m-1)|2(q^k-1)$. Let $\nu=(m,k)$ and $\lambda=\frac{q^m-1}{q^{\nu}-1}$. Then $\left(\lambda,\frac{q^k-1}{q^{\nu}-1}\right)=1$, and we have $\lambda|2$. If $\lambda=2$, then $m \ge 2$ and $\nu<m$, hence $\nu \le \frac{m}{2}$. We have $q^m-1=2(q^{\nu}-1) \le 2(q^{m/2}-1)$. This implies that $q \le q^{m/2}<2$, contradiction. So we must have $\lambda=\frac{q^m-1}{q^{\nu}-1}=1$, that is, $\nu=m$, hence $k=m$ or $2m$.

If $k=m$, this is equivalent to $d(r-1) \equiv 0 \pmod{r^2-1}$, that is $d \equiv 0 \pmod{r+1}$, and hence $(r+1)|(\triangle -2s)$. If $(r+1) \nmid (\triangle-2s)$, we must have $k=2m$.

(ii). $h_{d}(x) = h_{d'}(x)$ if and only if there exists an integer $k$, $1 \le k \le 2m$ such that $dq^k \equiv d' \pmod{r^2-1}$. Reducing the equation modulo $r-1$, by similar argument we find that $k=m$ or $2m$. If $k=2m$, then obviously $s \equiv s' \pmod{r+1}$. Otherwise $k=m$, we have $\left(s(r-1)+\triangle\right)r \equiv s'(r-1)+\triangle \pmod{r^2-1}$. This is equivalent to $\triangle \equiv s+s' \pmod{r+1}$ by simple computation. This completes the proof of Lemma \ref{3:lem33}.
\end{proof}

From Lemma \ref{3:lem33} we immediately obtain the following.

\begin{lemma} \label{3:lem2} (1). For $\ca$, let assumptions be as in Theorem \ref{1:thm1}. Then
\begin{itemize}
\item[(i).] $\deg h_{d_0}(x)=m$ and $\deg h_{d_i}(x)=2m, \forall \, 1 \le i \le t$.

\item[(ii).] $h_{d_i}(x) \ne h_{d_j}(x)$ for any $0 \le i \ne j \le t$.
\end{itemize}

\noindent (2). For $\cb$, let assumptions be as in Theorem \ref{1:thm2}. Then
\begin{itemize}
\item[(i).] $\deg h_{\widetilde{d}_i}(x)=2m, \forall \, 1 \le i \le t$.

\item[(ii).] $h_{\widetilde{d}_i}(x) \ne h_{\widetilde{d}_j}(x)$ for any $1 \le i \ne j \le t$.
\end{itemize}

\end{lemma}

%%%%%%%%%%%%%%%%%%%%%%%%%%%%%%%%%%%%%%%%%%%%%%%%%%
%%%%%%%%%%%%%%%%%%%%%%%%%%%%%%%%%%%%%%%%%%%%%%%%%%
\section{$p=2$: Proofs of (i) of Theorems \ref{1:thm1} and \ref{1:thm2}}\label{sec-3}
%%%%%%%%%%%%%%%%%%%%%%%%%%%%%%%%%%%%%%%%%%%%%%%%%%
%%%%%%%%%%%%%%%%%%%%%%%%%%%%%%%%%%%%%%%%%%%%%%%%%%

We first prove Statement (i) of Theorem \ref{1:thm1}. By Delsarte's Theorem \cite{dels} and Lemma \ref{3:lem2}, $\ca$ is a cyclic code of length $n$ with parity-check polynomial given by $\prod_{i=0}^t h_{d_i}(x)$, which is of degree $(2t+1)m$, hence $\ca$ has dimension $(2t+1)m$ over $\gf(q)$. Since
\[\delta=\left(d_0,\ldots,d_t,r^2-1\right)=\left((r+1)f,(r-1)e\right),\]
we see that the Hamming weight of a codeword $c(\vec{a})$ can be expressed as
\begin{eqnarray*} \delta \omega_H(c(\vec{a}))&=&r^2-\#\left\{x \in \gf(r^2): \tr_{r/q}\left(a_0x^{d_0}\right)+\tr_{r^2/q}\left(\sum_{j=1}^ta_jx^{d_j}\right)=0 \right\}\\
&=&r^2-\frac{1}{q}\sum_{x \in \gf(r^2)} \sum_{\lambda \in \gf(q)} \psi_q \left\{\lambda \tr_{r/q} \left( a_0x^{d_0}\right)+\lambda \tr_{r^2/q}\left( \sum_{j=1}^ta_jx^{d_j}\right)\right\} \\
&=&r^2\left(1-\frac{1}{q}\right)-\frac{S(\vec{a})}{q},
\end{eqnarray*}
where $\psi_q:\gf(q) \to \mathbb{C}^*$ is the standard additive character given by $\psi_q(x)=\zeta_p^{\tr_{q/p} (x)}$ for any $x \in \gf(q)$, $\zeta_p=\exp\left(2 \pi \sqrt{-1}/p\right)$, and
\begin{eqnarray} \label{3:sa} S(\vec{a}):=(q-1)+\sum_{\lambda \in \gf(q)^*}\sum_{x \in \gf(r^2)^*} \psi_q \left\{\lambda\tr_{r/q}\left(a_0x^{d_0}\right)+\lambda \tr_{r^2/q}\left(\sum_{j=1}^ta_jx^{d_j}\right)\right\}.\end{eqnarray}

Since $(r-1,r+1)=1$, we can write each $x \in \gf(r^2)^*$ uniquely as $x=yz$ for $y \in \gf(r)^*$ and $z \in U:=\{\omega \in \gf(r^2): \bar{\omega} \omega=\omega^{r+1}=1\}$. Here we denote $\bar{x}:=x^r$. Note that $U$ is a cyclic subgroup of $\gf(r^2)^*$ generated by $\gamma^{r-1}$. Since $y^r=y$ for any $y \in \gf(r)$, from (\ref{2:da1}) we have
\[x^{d_j}=y^{d_j}z^{d_j}=y^{2f} z^{-2jh}, \forall j, \]
and
\[x^{rd_j}=y^{2rf}z^{-2rjh}=y^{2f} z^{2jh}, \forall j. \]
Hence
\[S(\vec{a})=(q-1)+\sum_{z \in U}\sum_{y \in \gf(r)^*} \sum_{\lambda \in \gf(q)^*} \psi_r\left\{\lambda y^{2f}\left(a_0 +\sum_{j=1}^t a_j z^{-2jh}+\bar{a}_jz^{2jh}\right)\right\}.\]
Here $\psi_r: \gf(r) \to \{\pm 1\}$ is the standard additive character. Since $\left(2f,\frac{r-1}{q-1}\right)=\left(f,\frac{r-1}{q-1}\right)=1$, we observe that as $\lambda$ runs over $\gf(q)^*$ and $y$ runs over $\gf(r)^*$ respectively, the value $\lambda y^{2f}$ will run over each element of $\gf(r)^*$ exactly $(q-1)$ times. Hence we obtain
\begin{eqnarray*} \label{3:saq} S(\vec{a})=(q-1)+(q-1) \sum_{z \in U} \sum_{y \in \gf(r)^*} \psi_r \left\{ y \left( a_0 +\sum_{j=1}^t a_j z^{-2jh}+ \bar{a}_j
z^{2jh} \right) \right\}. \end{eqnarray*}
Clearly $S(\vec{a})=(q-1)r(N-1)$, where $N$ is the number of $z \in U$ such that
\[a_0+\sum_{j=1}^t a_j z^{-2jh}+\bar{a}_j z^{2jh}=0. \]
Letting $u=z^{-2h}$ and multiplying $u^{t}$ on both sides, we find
\[a_0u^t+\sum_{j=1}^ta_ju^{t+j}+\bar{a}_j u^{t-j}=0. \]
This is a polynomial of degree at most $2t$, it may have $0,1,\ldots$, or $2t$ solutions for $u$, and for each such $u \in U$, the number of $z \in U$ such that $z^{-2h}=u$ is always $e=(2h,r+1)=(h,r+1)$. Hence the possible values of $N$ are $je, \forall 0 \le j \le 2t$. This indicates that $S(\vec{a})$ and $\omega_H(c(\vec{a}))$ take at most $(2t+1)$ distinct values. This proves (i) of Theorem \ref{1:thm1}.

Statement (i) of Theorem \ref{1:thm2} can be proved similarly by using the above idea and by modifying the proof of \cite[Theorem 1]{XLZD} for $\cb$ accordingly. We omit the details. \qquad $\square$

%%%%%%%%%%%%%%%%%%%%%%%%%%%%%%%%%%%%%%%%%%%%%%%%%%
%%%%%%%%%%%%%%%%%%%%%%%%%%%%%%%%%%%%%%%%%%%%%%%%%%
\section{$p=2$: Proofs of (ii) of Theorems \ref{1:thm1} and \ref{1:thm2}}  \label{sec-4}
%%%%%%%%%%%%%%%%%%%%%%%%%%%%%%%%%%%%%%%%%%%%%%%%%%
%%%%%%%%%%%%%%%%%%%%%%%%%%%%%%%%%%%%%%%%%%%%%%%%%%

Since it is proved that there are only a few non-zero weights in $\ca$ and $\cb$, a standard procedure to determine the weight distribution is to compute power moment identities.

We now prove Statement (ii) of Theorem \ref{1:thm1}. Let $\mu_j$ be the frequency of weight $w_j$ for each $j$. Obviously $S(\vec{a})=(q-1)r^2$ if and only if $\vec{a}=\vec{0}$. We have
\begin{eqnarray} \label{4:id0} r^{1+2t}=1+\sum_{j=0}^{2t} \mu_j, \end{eqnarray}
and for any positive integer $k$,
\begin{eqnarray} \label{4:idr} \sum_{\substack{a_0 \in \gf(r)\\
a_j \in \gf(r^2), 1 \le j \le t}} \left(S(\vec{a})-(q-1)\right)^k=(q-1)^k\left(r^{2}-1\right)^k+\sum_{j=0}^{2t}  (q-1)^k\left(jer-r-1\right)^k \mu_j. \end{eqnarray}
On the other hand, by the orthogonal relation
\[\frac{1}{r^2} \sum_{x \in \gf(r^2)} \psi_q\left\{\tr_{r^2/q}(xa)\right\}=\left\{\begin{array}{ccl}
0&:& \mbox{ if } a \in \gf(r^2)^*\\
1&:& \mbox{ if } a =0, \end{array}\right.\]
we find easily that
\begin{eqnarray} \label{4:idr2} \sum_{\substack{a_0 \in \gf(r)\\
a_j \in \gf(r^2), 1 \le j \le t}} \left(S(\vec{a})-1\right)^k=r^{1+2t}M_k, \end{eqnarray}
where $M_k$ denotes the number of solutions $(\lambda_1,\ldots,\lambda_k) \in \left(\gf(q)^*\right)^k$ and $(x_1,\ldots,x_k) \in \left(\gf(r^2)^*\right)^k$ that satisfy the equations
\begin{eqnarray} \label{4:nra}
\left\{\begin{array}{ccc}
\lambda_1 x_1^{d_0}+\lambda_2x_2^{d_0}+\cdots+\lambda_kx_k^{d_0} &=&0, \\
\lambda_1x_1^{d_1}+\lambda_2x_2^{d_1}+\cdots+\lambda_kx_k^{d_1} &=&0, \\
\cdots \cdots & & \\
\lambda_1x_1^{d_t}+\lambda_2x_2^{d_t}+\cdots+\lambda_kx_k^{d_t} &=&0.
\end{array}\right.
\end{eqnarray}
Lemma \ref{4:lem1} which we will prove below states that $M_k=(q-1)^kN_k$ for any $1 \le k \le 2t$, where $N_k$ is given by the formula (\ref{2:nr}). Combining this with identities (\ref{4:id0}), (\ref{4:idr}) and (\ref{4:idr2}) for $1 \le k \le 2t$, we obtain the matrix equation
\[M_t^{(1)} \cdot \vec{\mu}=\vec{b}, \]
where $M_t^{(1)}, \vec{\mu}$ and $\vec{b}$ are explicitly defined in Theorem \ref{1:thm1}. Since $M_t^{(1)}$ is invertible, we obtain $\vec{\mu}=\left(M_t^{(1)}\right)^{-1} \cdot \vec{b}$, as claimed by (ii) of Theorem \ref{1:thm1}. Now we prove the technical lemma.

\begin{lemma} \label{4:lem1} $M_k=(q-1)^kN_k$ for any $1 \le k \le 2t$, where $N_k$ is given by the formula (\ref{2:nr}).
\end{lemma}
\begin{proof}

Using the same notation as before, we may write each $x_i \in \gf(r^2)^*$ as
\begin{eqnarray} \label{6:xi} x_i=y_i z_i, \quad y_i \in \gf(r)^*, z_i \in U. \end{eqnarray}
Since
\[x_i^{d_j}=y_i^{2f} z_i^{-2jh}, \, \forall i,j, \]
The equations (\ref{4:nra}) can be written as
\begin{eqnarray} \label{4:nraa} \sum_{i=1}^k \lambda_i \cdot y_i^{2f} z_i^{-2jh}=0, \forall 0 \le j \le t. \end{eqnarray}

Since $\left(2f,\frac{r-1}{q-1}\right)=1$, $\lambda_i \cdot y_i^{2f}$ takes each value of $\gf(r)^*$ exactly $(q-1)$ times as $\lambda_i$ and $y_i$ run over the sets $\gf(q)^*$ and $\gf(r)^*$ respectively. So $M_{k}=(q-1)^kM_{k,1}$ where $M_{k,1}$ counts the number of $y_i \in \gf(r)^*, z_i \in U \, \forall i$ such that
\begin{eqnarray} \label{4:nrb} \sum_{i=1}^k y_i z_i^{-2jh}=0, \forall 0 \le j \le t. \end{eqnarray}
In \cite{XLZD} we have used a combinatorial method to obtain the number of solutions to equations (\ref{4:nrb}). Roughly speaking, let $u_i=z_i^{-2h} \in U^e$ where $e=(2h,r+1)=(h,r+1)$. Using $y_i,u_i$'s, we can write (\ref{4:nrb}) as a matrix equation ${\bf A} \cdot \vec{y}=\vec{0}$ where
\begin{eqnarray*}
{\bf A}=\left[\begin{array}{rrrr}
1 & 1 & \cdots &1\\
u_1 & u_2 & \cdots &u_k \\
u_1^2 & u_2^2 & \cdots &u_k^2 \\
\vdots & \vdots & \ddots &\vdots \\
u_1^{t} &u_2^{t} & \cdots &u_k^{t}
\end{array} \right], \quad \vec{y}=
\left[\begin{array}{c}
y_1\\
y_2\\
\vdots\\
y_k\end{array}\right]. \end{eqnarray*}
We observe that ${\bf A}$ is a Vandermonde matrix. We may take the $r$-th power of each equation to obtain additional $(t-1)$ equations. It turns out that for any $1 \le k \le 2t$ there are only ``trivial'' solutions which can be counted exactly by using combinatorial argument. We conclude that $M_{k,1}=N_k$, which is given by the formula (\ref{2:nr}). Interested readers may review \cite{XLZD} for details. Therefore we obtain $M_k=(p-1)^kN_k$ as desired.

Statement (ii) of Theorem \ref{1:thm2} can be proved similarly, by using the above idea and by modifying the proof of \cite[Theorem 2]{XLZD} for $\cb$ accordingly. We omit the details. \end{proof}

%%%%%%%%%%%%%%%%%%%%%%%%%%%%%%%%%%%%%%%%%%%%%%%%%%
%%%%%%%%%%%%%%%%%%%%%%%%%%%%%%%%%%%%%%%%%%%%%%%%%%
\section{$p \ge 3$: Proofs of (i) of Theorems \ref{1:thm1} and \ref{1:thm2}}\label{sec-32}
%%%%%%%%%%%%%%%%%%%%%%%%%%%%%%%%%%%%%%%%%%%%%%%%%%
%%%%%%%%%%%%%%%%%%%%%%%%%%%%%%%%%%%%%%%%%%%%%%%%%%

We first prove Statement (i) of Theorem \ref{1:thm1}. Similar to the case that $p=2$, $\ca$ is a cyclic code of length $n$ with parity-check polynomial given by $\prod_{i=0}^t h_{d_i}(x)$, which is of degree $(2t+1)m$, hence $\ca$ has dimension $(2t+1)m$ over $\gf(q)$, and the Hamming weight of a codeword $c(\vec{a})$ can be expressed as
\begin{eqnarray*} \delta \omega_H(c(\vec{a}))&=&r^2\left(1-\frac{1}{q}\right)-\frac{S(\vec{a})}{q},
\end{eqnarray*}
where
\begin{eqnarray} \label{32:sa} S(\vec{a}):=(q-1)+\sum_{\lambda \in \gf(q)^*}\sum_{x \in \gf(r^2)^*} \psi_q \left\{\lambda\tr_{r/q}\left(a_0x^{d_0}\right)+\lambda \tr_{r^2/q}\left(\sum_{j=1}^ta_jx^{d_j}\right)\right\}.\end{eqnarray}

\subsection{Case 1: $m$ is odd.} We write each $x \in \gf(r^2)^*$ uniquely as $x=yw$ for $y \in \gf(r)^*$ and $w \in \Omega=\{ 1,\gamma,\gamma^2,\ldots,\gamma^r\}$ (see also \cite[Lemma 2]{gegeng2}). We may observe that $U:=\{\alpha^{r-1}: \alpha \in \Omega \}$ is a cyclic subgroup of $\gf(r^2)^*$ generated by $\gamma^{r-1}$. Since $y^r=y$ for any $y \in \gf(r)$, from (\ref{2:da1}) we have
\[x^{d_j}=y^{d_j}\omega^{d_j}=y^{2f} \omega^{d_j}, \forall j, \]
and
\[x^{rd_j}=y^{2fr}\omega^{rd_j}=y^{2f} \overline{\omega}^{d_j}, \forall j. \]
Here we denote $\bar{x}:=x^r$. Hence
\[S(\vec{a})=(q-1)+\sum_{\omega \in \Omega}\sum_{y \in \gf(r)^*} \sum_{\lambda \in \gf(q)^*} \psi_q \left\{\tr_{r/q}\left(\lambda y^{2f}\left\{a_0\omega^{d_0} +\sum_{j=1}^ta_j\omega^{d_j}+\bar{a}_j\overline{\omega}^{d_j}\right\}\right)\right\}.\]
Since $\frac{r-1}{q-1} \equiv m \pmod{2}$ is odd, $\left(2f,\frac{r-1}{q-1}\right)=\left(f,\frac{r-1}{q-1}\right)=1$. We observe that as $\lambda$ runs over $\gf(q)^*$ and $y$ runs over $\gf(r)^*$ respectively, the value $\lambda y^{2f}$ will run over each element of $\gf(r)^*$ exactly $(q-1)$ times. Hence we obtain
\begin{eqnarray*} \label{32:saq} S(\vec{a})=(q-1)+(q-1)\sum_{\omega \in \Omega}\sum_{y \in \gf(r)^*} \psi_q \left\{\tr_{r/q}\left(y\left\{a_0\omega^{d_0} +\sum_{j=1}^ta_j\omega^{d_j}+\bar{a}_j
\overline{\omega}^{d_j}\right\}\right)\right\}. \end{eqnarray*}
Clearly $S(\vec{a})=(q-1)r(N-1)$, where $N$ is the number of $\omega \in \Omega$ such that
\[a_0\omega^{d_0} +\sum_{j=1}^ta_j\omega^{d_j}+\bar{a}_j\omega^{rd_j}=0. \]
Dividing $\omega^{d_0}$ on both sides and writing $z=\omega^{r-1} \in U$, the equation becomes
\[a_0+\sum_{j=1}^ta_j z^{jh}+\bar{a}_jz^{-jh}=0. \]
Letting $u=z^{h}$ and multiplying $u^{t}$ on both sides, we find
\[a_0u^t+\sum_{j=1}^ta_ju^{t+j}+\bar{a}_j u^{t-j}=0. \]
This is a polynomial of degree at most $2t$, so possibly it may have $0,1,\ldots$, or $2t$ solutions for $u$, and for each such $u \in U$, the number of $z \in U$ such that $z^{h}=u$ is always $e=(h,r+1)$. Hence the possible values of $N$ are $je, \forall 0 \le j \le 2t$. This indicates that $S(\vec{a})$ and $\omega_H(c(\vec{a}))$ take at most $(2t+1)$ distinct values. This proves (i) of Theorem \ref{1:thm1} when $m$ is odd.

\subsection{Case 2: $m$ and $h$ are both even.} For this case we use a different strategy. Since
\[\gf(r)^* \bigcap U=\{-1,1\}, \]
we may write each $x \in \gf(r^2)^*$ as
\begin{eqnarray} \label{32:xsplit} x=y z \epsilon, \quad x \in \gf(r)^*,\, z \in U,\, \epsilon \in \{\xi,1\}, \end{eqnarray}
where $\xi \in \gf(r^2)^*$ is a fixed non-square, $U$ is the cyclic subgroup of $\gf(r^2)^*$ generated by $\gamma^{r-1}$. We may choose $\xi=\gamma^{(r+1)/2}$ as $r =q^m \equiv 1 \pmod{4}$. It is clear that as $y,z,\epsilon$ run over the sets $\gf(r)^*, U$ and $\{\xi,1\}$ respectively, the value $x$ will run over each element of $\gf(r^2)^*$ exactly twice. Also observing
\[x^{d_j}=(yz\epsilon)^{d_j}= \left(y^2\epsilon^{r+1}\right)^f z^{-2jh}\epsilon^{(r-1)jh}, \forall j,\]
\[x^{rd_j}=(yz\epsilon)^{rd_j}= \left(y^2\epsilon^{r+1}\right)^f z^{2jh}\epsilon^{-(r-1)jh}, \forall j,\]
and
\[\xi^{(r-1)h}=\left(\gamma^{(r^2-1)}\right)^{h/2}=1, \]
we can rewrite (\ref{32:sa}) as
\[S(\vec{a})=(q-1)+\frac{1}{2}\sum_{\lambda \in \gf(q)^*}\sum_{\substack{y \in \gf(r)^*\\
z \in U\\
\epsilon \in \{1,\xi\}}} \psi_q \left\{\tr_{r/q}\left(\lambda\left(y^2\epsilon^{r+1}\right)^f \left\{ a_0+
\sum_{j=1}^ta_jz^{-2jh}+\bar{a}_j z^{2jh}\right\} \right)\right\}.\]
Next we observe that $\xi^{r+1}=\left(\gamma^{r+1}\right)^{(r+1)/2}$ is a non-square in $\gf(r)^*$, so as $y$ runs over $\gf(r)^*$ and $\epsilon$ runs over $\{\xi,1\}$ respectively, the value $y^2 \epsilon^{r+1}$ will run over each element of $\gf(r)^*$ exactly twice. Therefore we obtain
\[S(\vec{a})=(q-1)+\sum_{z \in U} \sum_{\lambda \in \gf(q)^*} \sum_{y \in \gf(r)^*} \psi_q \left\{\tr_{r/q}\left(\lambda y^f \left\{ a_0+
\sum_{j=1}^ta_jz^{-2jh}+\bar{a}_j z^{2jh}\right\} \right)\right\}.\]
Since $\left(f,\frac{r-1}{q-1}\right)=1$, $\lambda y^f$ will take each value of $\gf(r)^*$ exactly $(q-1)$ times as $y,\lambda$ runs over $\gf(r)^*$ and $\gf(q)^*$ respectively. So
\begin{eqnarray*} S(\vec{a})&=&(q-1)+ (q-1) \sum_{z \in U} \sum_{y \in \gf(r)^*} \psi_q \left\{\tr_{r/q}\left(y \left\{ a_0+
\sum_{j=1}^ta_jz^{-2jh}+\bar{a}_j z^{2jh}\right\} \right)\right\}\\
&=&(q-1)r(N-1),\end{eqnarray*}
where $N$ is the number of $z \in U$ such that
\[a_0+\sum_{j=1}^ta_j z^{-2jh}+\bar{a}_jz^{2jh}=0. \]
Letting $u=z^{-2h}$ and multiplying $u^{t}$ on both sides, we find
\[a_0u^t+\sum_{j=1}^ta_ju^{t+j}+\bar{a}_j u^{t-j}=0. \]
This yields at most $2t$ solutions for $u$, and for each such $u$, the number of $z \in U$ such that $z^{-2h}=u$ is exactly $(-2h,r+1)=(h,r+1)=e$ because $r \equiv 1 \pmod{4}$ and $2|h$. Hence $N \in \{je: 0 \le j \le 2t\}$. This indicates again that $S(\vec{a})$ and $\omega_H(c(\vec{a}))$ take at most $(2t+1)$ distinct values. This concludes the case for both $m$ and $h$ being even. Now (i) of Theorem \ref{1:thm1} is proved.

Statement (i) of Theorem \ref{1:thm2} can be proved similarly by using the above idea and by modifying the proof of \cite[Theorem 1]{XLZD} for $\cb$ accordingly. We omit the details. \qquad $\square$

%%%%%%%%%%%%%%%%%%%%%%%%%%%%%%%%%%%%%%%%%%%%%%%%%%
%%%%%%%%%%%%%%%%%%%%%%%%%%%%%%%%%%%%%%%%%%%%%%%%%%
\section{$p \ge 3$: Proofs of (ii) of Theorems \ref{1:thm1} and \ref{1:thm2}}  \label{sec-42}
%%%%%%%%%%%%%%%%%%%%%%%%%%%%%%%%%%%%%%%%%%%%%%%%%%
%%%%%%%%%%%%%%%%%%%%%%%%%%%%%%%%%%%%%%%%%%%%%%%%%%

We now prove Statement (ii) of Theorem \ref{1:thm1}. Let $\mu_j$ be the frequency of weight $w_j$ for each $j$. Similar to the case $p=2$, we have
\begin{eqnarray} \label{42:id0} r^{1+2t}=1+\sum_{j=0}^{2t} \mu_j, \end{eqnarray}
and for any positive integer $k$,
\begin{eqnarray} \label{42:idr} \sum_{\substack{a_0 \in \gf(r)\\
a_j \in \gf(r^2), 1 \le j \le t}} \left(S(\vec{a})-(q-1)\right)^k=(q-1)^k\left(r^{2}-1\right)^k+\sum_{j=0}^{2t}  (q-1)^k \left(jer-r-1\right)^k \mu_j, \end{eqnarray}
and
\begin{eqnarray} \label{42:idr2} \sum_{\substack{a_0 \in \gf(r)\\
a_j \in \gf(r^2), 1 \le j \le t}} \left(S(\vec{a})-1\right)^k=q^{1+2t}M_k, \end{eqnarray}
where $M_k$ denotes the number of solutions $(\lambda_1,\ldots,\lambda_k) \in \left(\gf(q)^*\right)^k$ and $(x_1,\ldots,x_k) \in \left(\gf(r^2)^*\right)^k$ that satisfy the equations
\begin{eqnarray} \label{42:nra}
\left\{\begin{array}{ccc}
\lambda_1 x_1^{d_0}+\lambda_2x_2^{d_0}+\cdots+\lambda_kx_k^{d_0} &=&0, \\
\lambda_1x_1^{d_1}+\lambda_2x_2^{d_1}+\cdots+\lambda_kx_k^{d_1} &=&0, \\
\cdots \cdots & & \\
\lambda_1x_1^{d_t}+\lambda_2x_2^{d_t}+\cdots+\lambda_kx_k^{d_t} &=&0.
\end{array}\right.
\end{eqnarray}
Lemma \ref{43:lem1} which we will prove below states that that $M_k=(q-1)^kN_k$ for any $1 \le k \le 2t$, where $N_k$ is given by the formula (\ref{2:nr}). Combining this with identities (\ref{42:id0}), (\ref{42:idr}) and (\ref{42:idr2}) for $1 \le k \le 2t$, we obtain the matrix equation
\[M_t^{(1)} \cdot \vec{\mu}=\vec{b}, \]
where $M_t^{(1)}, \vec{\mu}$ and $\vec{b}$ are explicitly defined in Theorem \ref{1:thm1}. Since $M_t^{(1)}$ is invertible, we obtain $\vec{\mu}=\left(M_t^{(1)}\right)^{-1} \cdot \vec{b}$, as claimed by (ii) of Theorem \ref{1:thm1}. Now we prove the technical lemma.

\begin{lemma} \label{43:lem1} $M_k=(q-1)^kN_k$ for any $1 \le k \le 2t$, where $N_k$ is given by the formula (\ref{2:nr}).
\end{lemma}
\begin{proof}

Using the same notation as before, we may write each $x_i \in \gf(r^2)^*$ as
\begin{eqnarray} \label{6:xi} x_i=y_i z_i \epsilon_i, \quad y_i \in \gf(r)^*, z_i \in U, \epsilon_i \in \{\xi,1\}, \end{eqnarray}
where $\xi \in \gf(r^2)^*$ is a fixed non-square. As $y_i,z_i,\epsilon_i$ run over the sets $\gf(r)^*,U$ and $\{\xi,1\}$ respectively, $x_i=y_iz_i\epsilon_i$ will run over each element of $\gf(r^2)^*$ exactly twice. So $M_k=2^{-k}M_{k,1}$ where $M_{k,1}$ is the number of $\lambda_i \in \gf(q)^*, y_i \in \gf(r)^*,z_i \in U,\epsilon_i \in \{\xi,1\}, 1 \le i \le r$ such that $\lambda_i, x_i=y_iz_i\epsilon_i \, \forall i$ satisfy the equations (\ref{42:nra}) simultaneously. Since
\[x_i^{d_j}=\left(y_i^2\epsilon_i^{r+1}\right)^f \left(z_i^{-2}\epsilon_i^{r-1}\right)^{jh}, \, \forall j, \]
The equations (\ref{42:nra}) can be written as
\begin{eqnarray} \label{42:nraa} \sum_{i=1}^k \lambda_i \cdot \left(y_i^2\epsilon_i^{r+1}\right)^f \left(z_i^{-2}\epsilon_i^{r-1}\right)^{jh}=0, \forall 0 \le j \le t. \end{eqnarray}

\subsection{Case 1: $m$ is odd} Then $\frac{r-1}{q-1} \equiv m \equiv 1 \pmod{2}$, we may take $\xi=\gamma^{(r-1)/(q-1)}$. Then $\xi^{r+1} =\gamma^{(r^2-1)/(q-1)} \in \gf(q)^*$. For each fixed $\epsilon_i$, denoting $\lambda_i'=\lambda_i \epsilon_i^{(r+1)f} \in \gf(q)^*$, so to find $M_{k,1}$, it is equivalent to count the number of $\lambda_i' \in \gf(q)^*, y_i \in \gf(r)^*, z_i \in U, \epsilon_i \in \{\xi,1\} \, \forall i$ such that
\begin{eqnarray*} \sum_{i=1}^k \lambda_i' \cdot y_i^{2f} \left(z_i^{-2}\epsilon_i^{r-1}\right)^{jh}=0, \forall 0 \le j \le t. \end{eqnarray*}
Since $\left(2f,\frac{r-1}{q-1}\right)=1$, $\lambda_i' \cdot y_i^{2f}$ takes each value of $\gf(r)^*$ exactly $(q-1)$ times as $\lambda_i'$ and $y_i$ run over the sets $\gf(q)^*$ and $\gf(r)^*$ respectively. Moreover, $\xi^{r-1}=\left(\gamma^{r-1}\right)^{(r-1)/(q-1)}$ is a non-square in $U$, hence $z_i^{-2} \epsilon_i^{r-1}$ takes each value of $U$ exactly twice as $z_i$ and $\epsilon_i$ run over the sets $U$ and $\{\xi,1\}$ respectively. So $M_{k,1}=2^k(q-1)^kM_{k,2}$ where $M_{k,2}$ counts the number of $y_i \in \gf(r)^*, z_i \in U \, \forall i$ such that
\begin{eqnarray*} \label{42:nrb} \sum_{i=1}^k y_i z_i^{jh}=0, \forall 0 \le j \le t. \end{eqnarray*}
Similar to (\ref{4:nrb}), the above equations can be solved completely by using a combinatorial method of \cite{XLZD}. We conclude that $M_{k,2}=N_k$, which is given by the formula (\ref{2:nr}). Interested readers may review \cite{XLZD} for details. Therefore we obtain $M_k=(q-1)^kN_k$ as desired.

\subsection{Case 2: $m$ and $h$ are both even}

In this case $r \equiv 1 \pmod{4}$, we may take $\xi=\gamma^{(r+1)/2}$. Hence $\xi^{(r-1)h}=1$ as $2|h$ and $\xi^{r+1}$ is a non-square in $\gf(r)^*$. So $y_i^2 \epsilon_i^{r-1}$ takes each value of $\gf(r)^*$ exactly twice as $y_i$ and $\epsilon_i$ run over $\gf(r)^*$ and $\{\xi,1\}$. Hence (\ref{42:nraa}) can be reduced to
\begin{eqnarray*} \sum_{i=1}^k \lambda_i \cdot y_i^f z_i^{-2jh}=0, \forall 0 \le j \le t. \end{eqnarray*}
Since $\left(f,\frac{r-1}{q-1}\right)=1$, $\lambda_i y_i^f$ will take each value of $\gf(r)^*$ exactly $(q-1)$ times as $\lambda_i$ and $y_i$ run over $\gf(q)^*$ and $\gf(r)^*$ respectively. We have $M_{k,1}=2^k(q-1)^kM_{k,2}$, where $M_{k,2}$ is the number of solutions $y_i \in \gf(r)^*, z_i \in U \, \forall i$ such that
\begin{eqnarray*} \sum_{i=1}^k y z_i^{-2jh}=0, \forall 0 \le j \le t. \end{eqnarray*}
Again by using combinatorial argument as in \cite{XLZD} we can obtain that $M_{k,2}=N_k$ for any $1 \le k \le 2t$ which is given by (\ref{2:nr}). Hence we conclude $M_k=(q-1)^kN_k$ as desired. This completes the proof of Lemma \ref{43:lem1}.
\end{proof}

Statement (ii) of Theorem \ref{1:thm2} can be proved similarly by using the above idea and by modifying the proof of \cite[Theorem 2]{XLZD} for $\cb$ accordingly. We omit the details. \qquad $\square$

\section{Conclusions}\label{sec-conclusion}
In this paper we extended \cite{gegeng2,XLZD} further in two directions, that is, for any prime $p$, $q=p^l$ and $r=q^m$, we determined the weight distribution of the cyclic codes $\ca,\cb$ over $\gf(q)$ whose duals have $t+1$ and $t$ generalized Niho type zeroes respectively for any $t$ (see Theorems \ref{1:thm1} and \ref{1:thm2}). Numerical examples show that the classes considered contain many optimal linear codes which were not presented in \cite{gegeng2,XLZD}. %We also provide explicit formulas to compute their weight distribution.

\section*{Acknowledgement}
M. Xiong's research was supported by the Hong Kong Research Grants Council under Grant Nos. 609513 and 606211. N. Li's research was supported by the Norwegian Research Council. %Z. Zhou's research was supported by the Natural Science Foundation of China under Grant No. 61201243, and also the Application Fundamental Research Plan Project of Sichuan Province under Grant No. 2013JY0167.  C. Ding's research was supported by The Hong Kong Research Grants Council under Grant No. 600812.

\end{document}